\documentclass[11pt]{article}

\usepackage{amsmath,amssymb,fullpage,times,color,graphicx,subfigure}
\usepackage{algorithm, algorithmic}
\usepackage[toc,page]{appendix}
\usepackage{tikz}
\usepackage{enumitem}
\newtheorem{theorem}{Theorem}[section]

\newtheorem{lemma}[theorem]{Lemma}

\newtheorem{definition}[theorem]{Definition}

\newcommand{\qed}{\rule{2mm}{2mm}}
\newenvironment{proof}{\par\noindent{\bf Proof.}\quad}{  $\qed$}

\newcommand{\bool}{\{0,1\}}
\newcommand{\imply}{\Rightarrow}

\title{Complexity of Traveling Tournament Problem with Trip Length More Than Three}
\author{  
	Diptendu Chatterjee \thanks{
		Indian Statistical Institute,Kolkata, India.
		{\tt diptendu\_r@isical.ac.in}} 
}
\begin{document}
	
	\maketitle
	
	\begin{abstract}
		The Traveling Tournament Problem is a sports-scheduling problem where the goal is to minimize the total travel distance of teams playing a double round-robin tournament. The constraint $k$ is an imposed upper bound on the number of consecutive home or away matches. It is known that TTP is NP-Hard for $k=3$ as well as $k=\infty$. In this work, the general case has been settled by proving that TTP-$k$ is NP-Complete for any fixed $k>3$.  
	\end{abstract}

\section{Introduction}

Sports tournaments are very popular events all over the world. A huge amount of money is involved in organizing these tournaments and also a lots of revenue is generated by selling the tickets and broadcasting rights. Scheduling the matches is a very important aspect of these tournaments. Scheduling specifies the sequence of matches played by each participating team along with the venues. In a Double Round-robin Tournament, every pair of participating teams play exactly two matches between them once in both of their home venues. The Traveling Tournament Problem (TTP) asks for a double round robin schedule minimizing the total travel distance of the participating teams. The fairness condition imposes an upper bound $k$ to the maximum number of consecutive home or away matches played by any team.

TTP was first introduced by Easton, Nemhauser, and Trick~\cite{easton2001traveling}. The problem bears some similarity with Traveling Salesman Problem (TSP). In fact, a reduction of the unconstrained version of TTP (or $k=\infty$) from TSP has been shown by Bhattacharyya \cite{bhattacharyya2016complexity} proving the basic problem to be NP hard.  When the maximum permissible length consecutive home or away matches is set to $3$, TTP is proven to be NP-Hard as well by Thielen and Westphal \cite{thielen2011complexity}. In this work, the natural question has been asked, \emph{Is TTP NP-Hard for any fixed $k>3$}? 

\subsection{Problem Definition}
Let $T$ be the set of teams with $|T|=n$. Let $\Delta$ be a square matrix of dimension $n \times n$ whose element at $i^{th}$ row and $j^{th}$ column corresponds to the distance between the home venues of $i^{th}$ and $j^{th}$ team in $T$ for all $i,j\leq|T|$. $\Delta$ is symmetric with diagonal terms $0$ and all the distances in $\Delta$ satisfy triangle inequality. 
\begin{definition}
	\textbf{Decision Version of TTP-$k$:} For a fixed a natural number $k$, an even integer $n$, a given set of teams $T$ with $|T|=n$, mutual distances matrix ($\Delta$) and a rational number $\delta$, is it possible to schedule a double round-robin tournament such that the total travel distance of the tournament is less or equal to $\delta$, where no team can have more than $k$ consecutive away matches or consecutive home matches and no team plays its consecutive matches with same opponent. 
\end{definition}

\subsection{Previous Work}
The Traveling Tournament Problem was first introduced by Easton, Nemhauser, and Trick~\cite{easton2001traveling}. Since then most of works focused on finding approximation algorithms or heuristics when $k=3$ \cite{benoist2001lagrange,easton2002solving,anagnostopoulos2006simulated,di2007composite, miyashiro2012approximation}. Thielen and Westphal \cite{thielen2011complexity} proved NP hardness for TTP-3. In a series of two papers \cite{imahori2010approximation,imahori20142} Imahori, Matsui, and Miyahiro showed approximation algorithms for the unconstrained TTP. Bhattacharyya \cite{bhattacharyya2016complexity} complemented this by proving the NP-hardness of unconstrained TTP.  For other values of $k$, only upper bound results are known. Approximation for TTP-2 is done by Thielen and Westphal \cite{thielen2010approximating} and improved by Xiao and Kou \cite{xiao2016improved}.  For TTP-$k$ with $k>3$, approximation algorithms are given by Yamaguchi et al. \cite{yamaguchi2011improved} and Westphal and Noparlik \cite{westphal20145}.

Many different problems related to sports scheduling have been thoroughly analyzed in the past. For detail surveys and study on round-robin tournament scheduling, the readers are referred to \cite{kendall2010scheduling,rasmussen2008round,trick1999challenge}. Graph Theoretic approach for solving scheduling problems can be found in \cite{de1981scheduling,de1988some}. 

\subsection{Approach towards the Problem}
In this work, generalization of the approach by Thielen and Westphal \cite{thielen2011complexity} has been done, who showed the NP Hardness of TTP-$3$. Like them, a reduction from the satisfiability problem has been shown. While \cite{thielen2011complexity} showed reduction from $3$-SAT, a reduction from $k$-SAT is shown here. However the reduction shown here is different in few crucial aspects. Firstly, the construction of the reduced TTP instance graph is different from that in~\cite{thielen2011complexity}. In order to accommodate trips of length $k$, a new graph in terms of vertices and edge weights is required. In addition, the trip structures and their synchronous assembly to get the tournament schedule is different from that of TTP-$3$. The reconstruction of $k$-SAT with specific properties of clauses require different technique. 

\subsection{Result}
Here the main theorem is the following.

\begin{theorem}
	\label{thm:main}
	TTP-$k$ for a fixed $k>3$ is strongly NP-Complete.
\end{theorem}

\section{Proof of Theorem~\ref{thm:main}}
This reduction requires that the input instance of the satisfiability problem satisfies certain properties. The first step is to show that \emph{any} input instance can be transformed into an instance with properties required for the reduction. Following notations are used throughout the paper. If $x\in\bool$ is a boolean variable, $\bar{x}$ denotes its complement, $\bar{x}=1\oplus x$.

\begin{lemma}
	\label{thm:satmod}
	For a $k$-SAT instance $F_k$ with $t$ variables and $p$ clauses there exists another $k$-SAT instance $F'_k$ with $t'$ variable and $p'$ clauses such that a satisfying assignment of $F_k$ implies a satisfying assignment of the variables of $F'_k$ and vice-versa. Moreover, the number of occurrence of all the $t'$ variables and there compliments are equal in $F'_k$, $t'\leq \left(t+\frac{k+1}{2}\right)$ and $k\mid p'$ with $k,t,t',p,p' \in \mathbb{N}$.
\end{lemma}

\begin{proof}
	Let $x_i$ and $\bar{x}_i $ ($i \in \{1,\dots,t\}$) be the variables in $F_k$. Suppose $n_i$ and $\bar{n}_i$ are the number of occurrence of $x_i$ and $\bar{x}_i$ in $F_k$ respectively. Without loss of generality it can be assume that $n_i \leq \bar{n}_i$. To make $n_i=\bar{n}_i$, few clauses has been added to $F_k$ depending on the value of $k$ is even or odd.
	\begin{enumerate}
		\item When $k$ is odd, clauses of the form $(x_i \lor x_{t+1} \lor \bar{x}_{t+1} \lor \dots \lor x_{t+\frac{k-1}{2}} \lor \bar{x}_{t+\frac{k-1}{2}})$ are added until $n_i=\bar{n}_i$ $\forall i \in \{1,\dots,t\}$ by introducing $(k-1)/2$ new variables. After adding these clauses number of clauses is even due to \textit{Handshaking Lemma}. Then by adding at most $(k-1)/2$ pairs of clauses of the form $(x_{t+1} \lor \bar{x}_{t+1} \lor \dots \lor x_{t+\frac{k-1}{2}} \lor \bar{x}_{t+\frac{k-1}{2}} \lor x_{t+\frac{k+1}{2}})$ and $(x_{t+1} \lor \bar{x}_{t+1} \lor \dots \lor x_{t+\frac{k-1}{2}} \lor \bar{x}_{t+\frac{k-1}{2}} \lor \bar{x}_{t+\frac{k+1}{2}})$, number of clauses can be made divisible by $k$ keeping $n_i=\bar{n}_i, \forall i \in \{1,\dots,t+\frac{k+1}{2}\}$.
		\item When $k$ is even, then using \textit{Handshaking Lemma} it can be said that there exist even number of indices $j \in \{1,\dots,t\}$ such that $(n_j+\bar{n}_j)$ is odd. So, $\|n_j-\bar{n}_j\|$ is odd. Without loss of generality, it can be assumed that $\bar{n}_j>n_j$. By identifying two indices of this kind, namely $m$ and $n$, clauses of the form $(x_m \lor x_n \lor x_{t+1} \lor \bar{x}_{t+1} \lor \dots \lor x_{t+\frac{k}{2}-1} \lor \bar{x}_{t+\frac{k}{2}-1})$ are added until $\|n_i-\bar{n}_i\|$ is even $\forall i \in \{1,\dots,t\}$ by introducing $(k/2-1)$ new variables. Now $\forall j \in \{1,\dots,t\}$ such that $n_j \neq \bar{n}_j$, pair of clauses of the forms $(x_j \lor x_{t+1} \lor \bar{x}_{t+1} \lor \dots \lor x_{t+\frac{k}{2}-1} \lor \bar{x}_{t+\frac{k}{2}-1} \lor x_{t+\frac{k}{2}})$ and $(x_j \lor x_{t+1} \lor \bar{x}_{t+1} \lor \dots \lor x_{t+\frac{k}{2}-1} \lor \bar{x}_{t+\frac{k}{2}-1} \lor \bar{x}_{t+\frac{k}{2}})$ are added until $n_j=\bar{n}_j$ $\forall j \in \{1,\dots,t\}$. Then by adding at most $(k-1)$ clauses of the form $(x_{t+1} \lor \bar{x}_{t+1} \lor \dots \lor x_{t+\frac{k}{2}} \lor \bar{x}_{t+\frac{k}{2}})$, the number of clauses can be made divisible by $k$ keeping $(n_i=\bar{n}_i)$ $\forall i \in \{1,\dots,t+\frac{k}{2}\}$.
	\end{enumerate}
	
	Let the resulting \textit{k-SAT} problem expression be $F'_k$ with $t'$ variables and $p'$ clauses where, $k\vert p'$. In both the cases explained above, all the additional clauses always give truth values. So, $F'_k$ will have a truth assignment of variables $x_i$ for all $i \in \{1,\dots,t'\}$ if and only if $F_k$ have a truth assignment of variables $x_i$ $\forall i \in \{1,\dots,t\}$. This proves the lemma.
\end{proof}

\subsection{TTP-$k$ is NP-Complete}
The first step is to show that TTP-$k$ is indeed in NP.
\begin{lemma}
	TTP-$k$ is in NP.
\end{lemma}

\begin{proof}
	In a decision version of \textit{TTP-$k$} with given set $T$ of $n$ teams, and the constraint $k$ it is verifiable in $O(n^2)$ whether a schedule is valid and gives a total travel distance less than $\delta$ or not. This ensures the membership of \textit{TTP-$k$} in \textit{NP}.
\end{proof} 

Next, a reduction from \textit{$k$-SAT} to \textit{TTP}-$k$ has been shown. For this, a special weighted graph $G=(V,E)$ has been constructed with one or more teams situated at each vertex of $G$ and a predefined value $\delta$ of total travel distance such that  there is a satisfying assignment of variables for \textit{k-SAT} if and only if there is a \textit{TTP}-$k$ schedule between the teams in $G$ with total travel distance less than $\delta$. First the input \textit{$k$-SAT} problem is modified using Lemma~\ref{thm:satmod} such that the resulting formula has the following properties.
\begin{enumerate}[label=(\alph*)]
	\item There are $t$ variables $x_1, x_2,\dots, x_t$.
	\item Number of occurrence of $x_i$ is equal to the number of occurrence of $\bar{x}_i$; $n_i=\overline{n}_i$.
	\item Number of clauses $p$ is divisible by $k$, $k\vert p$.
\end{enumerate}

\subsubsection{The Construction}

We start with the construction of the reduced instance graph $G$.  Recall that $k$ is the upper bound of number of consecutive home or away matches, and $n_i$ is the number of occurrence of the variable $x_i$. The main part of the graph is the union of $t$ many sub-graphs $G_1,G_2,\cdots,G_t$ where $t$ is the number of variables in the (modified) input SAT instance. Each $G_i$ consists of $(k+1)n_i$ vertices.

\begin{enumerate}[label=(\alph*)]
	\item $n_i$ many vertices are denoted by $x_{i,j}$ and $n_i$ many vertices are denoted by $\bar{x}_{i,j}$ where $ j \in \{1,\dots,n_i\}$.
	\item For $ j \in \{1,\dots,n_i\}$, the vertices $y_{i,j}$ denote $n_i$ many vertices.
	\item $ j \in \{1,\dots,n_i\}$ and $l\in \{1,\cdots,k-2\}$ the vertices $w_{i,j}^l$ are the remaining $(k-2)n_i$ many vertices.  
\end{enumerate}
In addition, there are $(k-1)p+1$ many vertices in the graph where $p$ is the number of clauses of the (modified) input SAT instance. There is a central vertex $v$.  Then there are $p$ vertices, $C_m$ and $(k-2)p$ vertices, $z_m^l \ \forall l \in \{1,\dots,k-2\}$ and $\forall m \in \{1,\dots,p\}$.

For ease of explanation, we summarize the important parameters related to $G$.  The total number of vertices in $G$ is $\left[\left(\sum_{i=1}^t(k+1)n_i\right)+(k-1)p+1\right]$. We also know that, $\sum_{i=1}^{t}2n_i=kp$.  So, The total number of vertices other than $v$ in $G$ is $\left(\frac{k(k+1)}{2}+k-1\right)p$ which we denote by $a$.

\begin{equation*}
	a=\left(\frac{k(k+1)}{2}+k-1\right)p=p\left(\frac{k^2+3k-2}{2}\right)
\end{equation*}

\textsc{Weights of the edges}. Let $M=\theta(a^5)$. First weight $M$ is assigned to the edges from $v$ to the vertices $x_{i,j}$ and $\bar{x}_{i,j}$, $(M-2)$ to the edges from $v$ to $y_{i,j}$ and $(M-2k+4)$ to edges from $v$ to $w^{k-2}_{i,j}$ of $G_i$ for every $j\in\{1,2,\dots,n_i\}$. Then $x_{i,j}$ is connected with $w^{k-2}_{i,j}$ through $k-3$ vertices serially namely $w^{1}_{i,j},w^{2}_{i,j},\dots,w^{k-3}_{i,j}$ where each of these consecutive vertices in this serial connection is connected with each other with an edge of weight $2$. Then $\bar{x}_{i,j}$ is connected to $w^{1}_{i,q}$ with an edge of weight $2$ where $q=j(mod \ n_i)+1$ and also $y_{i,j}$ is connected to both $x_{i,j}$ and $\bar{x}_{i,j}$ with edges of weight $2$ each as described in Figure \ref{subgraph_1}.

For the connection between the remaining vertices in $G$, first $x_{i,j}$ or $\bar{x}_{i,j}$ are connected to $c_m$ with an edge of weight $2$ if $x_i$ or $\bar{x}_{i,j}$ has its $j^{th}$ occurrence in the $m^{th}$ clause of the modified \textit{k-SAT}. Then $c_m$ is connected with $z^{k-2}_m$ through $k-3$ vertices serially namely $z^1_m,z^2_m,\dots,z^{k-3}_m$ where each of these consecutive vertices in this serial connection is connected with each other with an edge of weight $2$. At last $z^{k-2}_m$ and $c_m$ is connected to $v$ by edges of weights $(M-2k+6)$ and $(M+2)$ respectively as described in Figure \ref{subgraph_2}.

Formally, weights of all the edges of $G$ are listed as follows:

\begin{itemize}
	\item Weight$(w_{i,j}^r, w_{i,j}^s)$=2$\|r-s\|$ for all $i \in \{1,\dots,t\}$, for all $j \in \{1,\dots,n_i\}$, for all $r,s \in \{1,\dots,k$-2$\}$.
	\item Weight$(x_{i,j}, w_{i,j}^{s})$=2s for all $i \in \{1,\dots,t\}$, for all $j \in \{1,\dots,n_i\}$, for all $s \in \{1,\dots, k$-2$\}$.
	\item Weight$(\bar{x}_{i,j},w_{i,q}^{s})$=2s for all $i \in \{1,\dots,t\}$, for all $j \in \{1,\dots,n_i\}$, for all $s \in \{1,\dots, k$-2$\}$ and $q=j(mod \ n_i)+1$.
	\item Weight$(x_{i,j}, y_{i,j})$=Weight$(\bar{x}_{i,j}, y_{i,j})$=2 for all $i \in \{1,\dots,t\}$, for all $j \in \{1,\dots,n_i\}$.
	\item Weight$(c_m,z_m^r)$=2r for all $r \in \{1,\dots,k-2\}$ and for all $m \in \{1,\dots,p\}$.
	\item Weight$(z_m^r,z_m^s)$=2$\|r-s\|$ for all $r,s \in \{1,\dots,k-2\}$ and for all $m \in \{1,\dots,p\}$.
	\item Weight$(x_{i,j},c_m)$=2, if $j^{th}$ occurrence of $x_i$ is present in $m^{th}$ clause of the given \textit{k-SAT} expression.
	\item Weight$(\bar{x}_{i,j},c_m)$=2, if $j^{th}$ occurrence of $\bar{x}_i$ is present in $m^{th}$ clause of the given \textit{k-SAT} expression.
	\item Weight$(v,w_{i,j}^{k-2})$=M-2k+4 for all $i \in \{1,\dots,t\}$, for all $j \in \{1,\dots,n_i\}$.
	\item Weight$(v,y_{i,j})$=M-2 for all $i \in \{1,\dots,t\}$, for all $j \in \{1,\dots,n_i\}$.
	\item Weight$(v,x_{i,j})$=M for all $i \in \{1,\dots,t\}$, for all $j \in \{1,\dots,n_i\}$.
	\item Weight$(v,\bar{x}_{i,j})$=M for all $i \in \{1,\dots,t\}$, for all $j \in \{1,\dots,n_i\}$.
	\item Weight$(v,z_m^{k-2})$=M-2k+6 for all $m \in \{1,\dots,p\}$.
	\item Weight$(v,c_m)$=M+2 for all $m \in \{1,\dots,p\}$.
\end{itemize}  

All other edges in the complete graph $G$ is given maximum possible weights without violating triangle inequality.

\noindent\textsc{Creating TTP-$k$ instance.} Now the teams are placed on the vertices of $G$ to construct the reduced instance. 
\begin{itemize}
	\item Total number of teams is equal to $a^3+a$.
	\item At each vertex of $G$ except $v$, only one team is placed. This set of vertices or teams is denoted as $U$.
	\item $a^3$ teams are situated at $v$ and distance between them is $0$ and call this set of vertices or teams $T$. 
\end{itemize}

\begin{figure}[htb]
	\begin{center}
		\includegraphics[height=5in,width=4in,angle=0]{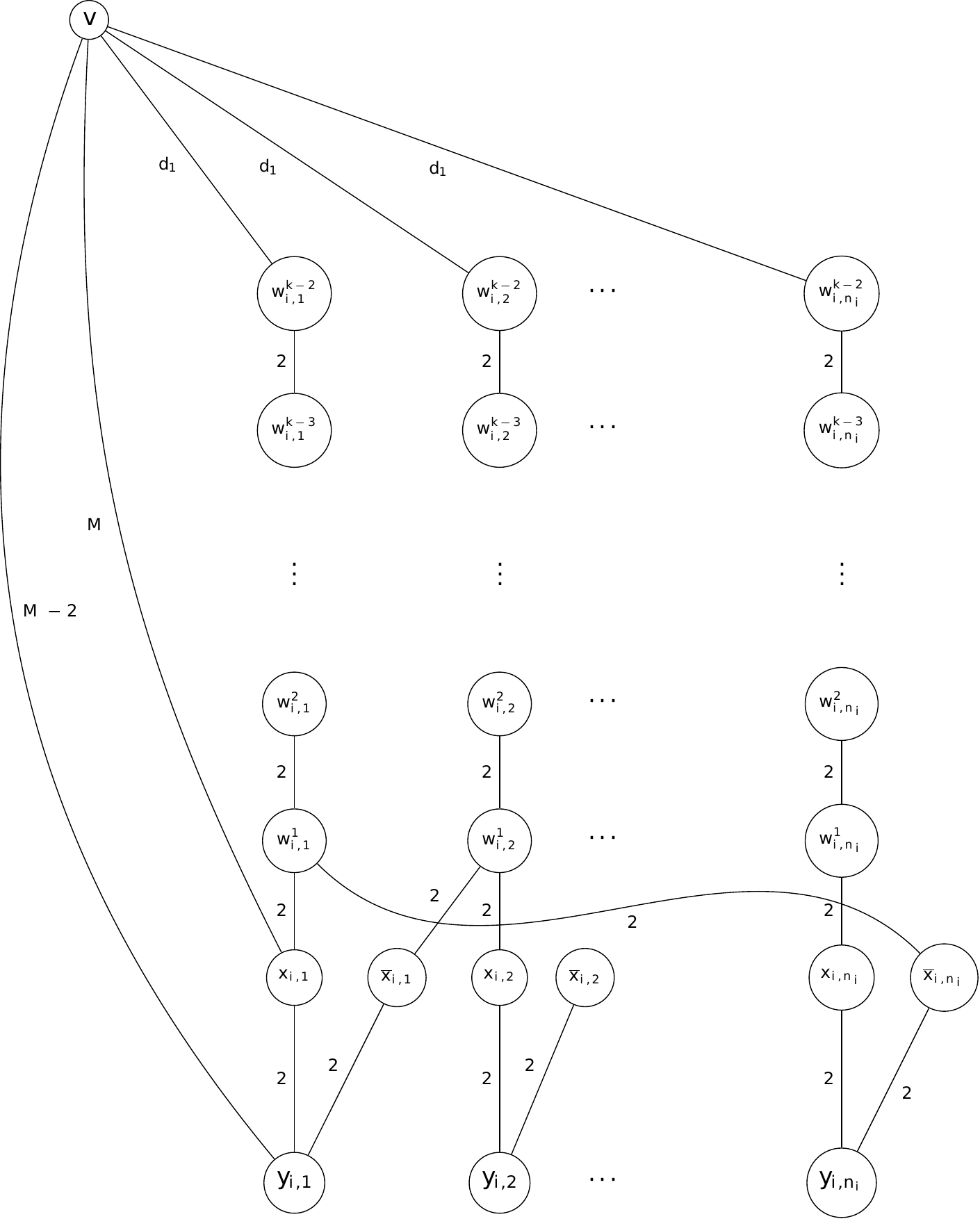}
		\caption{Sub-graph of $G$ for $i^{th}$ variable where $d_1=M-2k+4$}
		\label{subgraph_1}
	\end{center}
\end{figure}

\begin{figure}[htb]
	\begin{center}
		\includegraphics[height=5in,width=4in,angle=0]{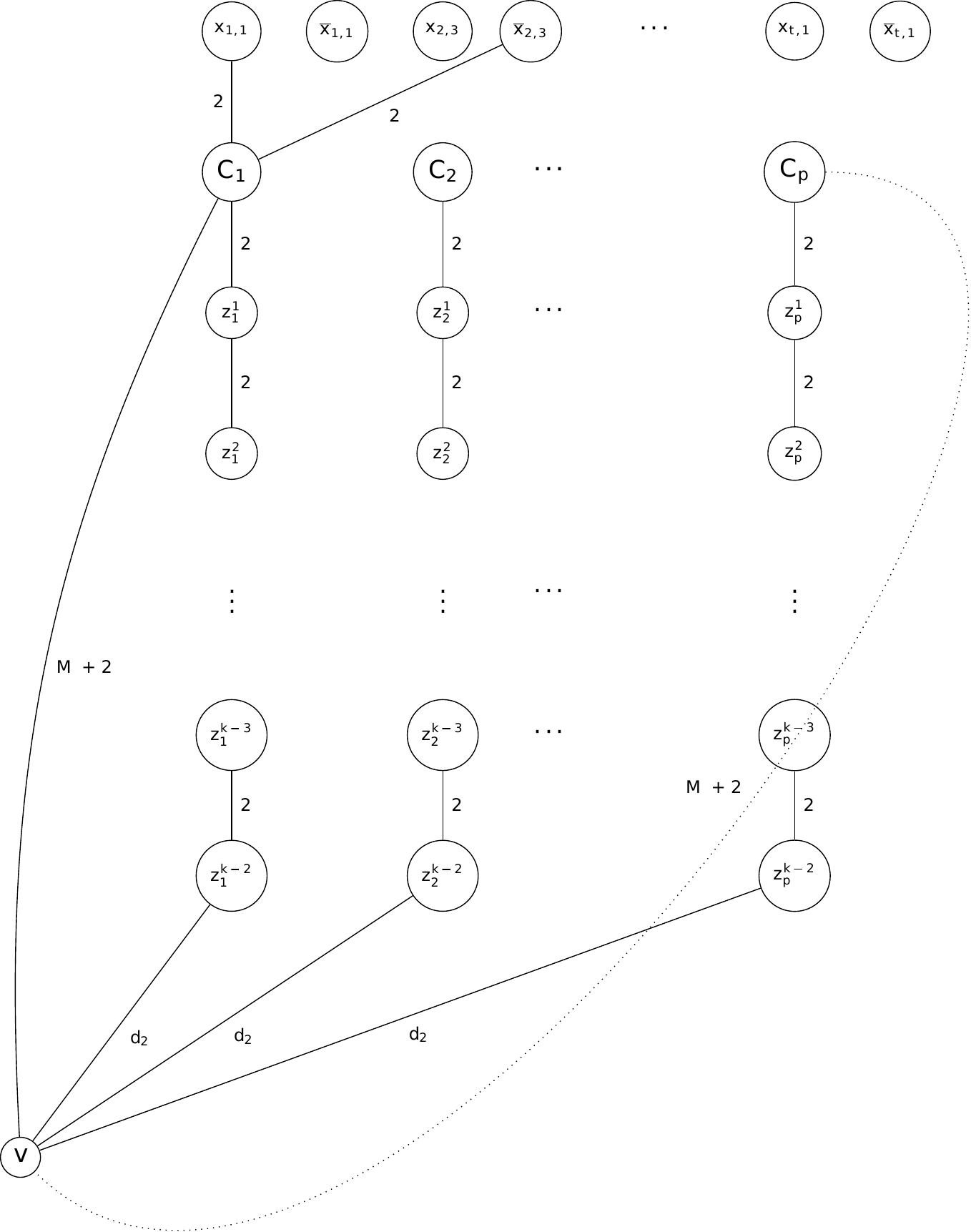}
		\caption{Sub-graph Corresponding to First Clause of the form $C_1=(x_1 \lor \bar{x}_2 \lor \dots)$ where $d_2=M-2k+6$}
		\label{subgraph_2}
	\end{center}
\end{figure}

We fix,
\begin{equation*}
	\delta = M\left[\frac{4a^4}{k}\right]+pa^3(k^2-3k+6)+2a(a-1)k+\frac{4a^4}{k}
\end{equation*}	

\begin{lemma}
	Weights of the edges of $G$ preserve the triangle Inequality.
\end{lemma}

\begin{proof}
	For a tuple $(v,x_{i,j},y_{i,j})$ or $(v,\bar{x}_{i,j},y_{i,j})$ for all $i \in \{1,\dots,t\}$ and $j \in \{1,\dots,n_i\}$, the triangle inequality is preserved as,	
	\begin{eqnarray*}
		Weight(v,x_{i,j})=M=Weight(x_{i,j},y_{i,j})+Weight(v,y_{i,j})=2+(M-2)\\
		Weight(v,\bar{x}_{i,j})=M=Weight(\bar{x}_{i,j},y_{i,j})+Weight(v,y_{i,j})=2+(M-2) 
	\end{eqnarray*}	
	
	For a tuple $(v,x_{i,j},w^{1}_{i,j})$ or $(v,\bar{x}_{i,j},w^{1}_{i,j})$ for all $i \in \{1,\dots,t\}$ and $j \in \{1,\dots,n_i\}$, the triangle inequality is preserved as,	
	\begin{eqnarray*}
		Weight(v,x_{i,j})=M=Weight(x_{i,j},w^{1}_{i,j})+Weight(v,w^{1}_{i,j})=2+(M-2)\\
		Weight(v,\bar{x}_{i,j})=M=Weight(\bar{x}_{i,j},w^{1}_{i,j})+Weight(v,w^{1}_{i,j})=2+(M-2) 
	\end{eqnarray*}
	
	For a tuple $(v,x_{i,j},c_m)$ or $(v,\bar{x}_{i,j},c_m)$ where $j^{th}$ occurrence of $x_i$ or $\bar{x}_i$ is present in $m^{th}$ clause of the given \textit{k-SAT} expression, the triangle inequality is preserved as,	
	\begin{eqnarray*}
		Weight(v,c_m)=M+2=Weight(x_{i,j},c_m)+Weight(v,x_{i,j})=2+M\\
		Weight(v,c_m)=M+2=Weight(\bar{x}_{i,j},c_m)+Weight(v,\bar{x}_{i,j})=2+M 
	\end{eqnarray*}
	
	The triangle inequality for all other tuple of three vertices in $G$ is followed from these above three cases as the weights are given maximum possible values without violating triangle inequality while assigned.
\end{proof}

\subsection{The Reduction}
As a desired value of $\delta$ has been got, now the only remaining part is the reduction. First, it has been shown that a given satisfying assignment of variables in a \textit{k-SAT} $\imply$ a \textit{TTP-k} schedule of total travel distance less than $\delta$. First, the tours for the $a^3$ vertices situated at $v$ are constructed and also showed that, these tours are so cheap in terms of travel distance that tours of similar structure must be there in a tournament where the total travel distance is desired to be as low as possible. So, $\frac{a}{k}$ node disjoint tours are constructed for a team at $v$ where, all the vertices $x_{i,j}, \bar{x}_{i,j}, y_{i,j}, w^r_{i,j}, c_m, z^r_m$ are visited. First, observe that as there is a satisfying assignment of variables in the \textit{k-SAT}. Let us define two conditions denoting with the value of a variable $b_m$ in the following manner,
\begin{equation*}
	b_m=1 \implies \exists i \in \{1,\dots,t\}\mbox{ such that } x_i=1 \ \& \ x_i \ \mbox{appears in the} \ m^{th} \mbox{clause}.
\end{equation*} 
\begin{equation*}
	b_m=0 \implies \exists i \in \{1,\dots,t\} \mbox{ such that } x_i=0 \ \& \ \bar{x}_i \ \mbox{appears in the} \ m^{th} \mbox{clause}.
\end{equation*}

For all $m \in \{1,\dots,p\}$, if $b_m=1$ then Weight$(x_{i,j},c_m)=2$. Now, to visit a team at $c_m$ only, a team in $v$ has to travel $2(M+2)$ distance. But if it travels to a vertex $x_{i,j}$, then to $c_m$ and travel through $z_m^1$ to $z_m^{k-2}$ and return to $v$, it travels the same distance, $2(M+2)$, and also visits $k$ vertices in a single trip which is a desired situation here. Afterwards in another trip to a vertex $\bar{x}_{i,j}$, in spite of directly traveling to and fro to $\bar{x}_{i,j}$ only, i.e. $2M$, if it travels to $w_{i,q}^{k-2}$  first then through all $w_{i,q}^s$ for $q=j(mod \ n_i)+1$ and $s \in \{1,\dots,k-3\}$ to $\bar{x}_{i,j}$ and $y_{i,j}$ and returns back to $v$, the travel distance is same as $2M$ and also $k$ vertices will be visited in a single trip. Multiple trip to these extra $(k-1)$ vertices would cost much more in comparison. Similar tours are taken when $b_m=0$ only interchanging $x_{i,j}$ and $y_{i,j}$ with $\bar{x}_{i,q}$ and  $y_{i,q}$ respectively. This leaves $(k-2)p/2$ number of $x_{i,j}$ type vertices to visit. As $k\vert p$, this can be done in $p(k-2)/2k$ trips each of length $k$ and travel distance less than $2(M+k(k-1))$. So the total distance traveled by a team situated at $v$ is upper bounded by $\delta_1$. Where,

\begin{eqnarray*}
	\delta_1&&=2(M+2)p+Mkp+M(k-2)p/k+(k-1)(k-2)p\\ &&=p\left[M\left(k+3-\frac{2}{k}\right)+k^2-3k+6\right]
\end{eqnarray*}  

With tours involving distance $M$ is minimized by covering exactly $k$ vertices in each of the tours. Now these tours can be numbered from $1$ to $\frac{a}{k}$ and vertices in each tour can be numbered from $1$ to $k$. So, all the $a$ vertices of $U$ are named as $u_{i,j}$ for all $i \in \{1,\dots,k\}$ and $j \in \{1,\dots,\frac{a}{k}\}$ such that $ u_{i,j}$ is the $i^{th}$ visited team of $j^{th}$ tour. Also the vertices in $T$ are partitioned in $a^2$ disjoint sets $T_1, T_2, \dots, T_{a^2}$ each of size $a$. Moreover $T_q=\{t_{r,q}$ such that $ r \in \{1,\dots,a\}\}$. The tours by the teams in $T$ are now designed in such a way that, $t_{1,1}$ will take tour number $1$, i.e. travel through $u_{1,1},u_{2,1},\dots,u_{k,1}$. Then $u_{1,1},u_{2,1},\dots,u_{k,1}$ visit $t_{1,1}$ in the same order. Similarly, for all $i \in \{2,\dots,k\}, t_{i,1}$ follows the same tour and visited by the same teams as $t_{1,1}$ but with a time delay of $(i-1)$. This way all the teams in $T_1$ first complete visits to all the teams in tour $1$ and then get visited by them. Then $T_1$ starts tour $2$. This way teams in $T_1$ plays with the teams in tour $1$ to tour $\frac{a}{2k}$ in such a way that the teams in $T_1$ visit first and then they get visited. Similarly, teams in $T_q$ for all $q \in \{1,\dots,\frac{a^2}{2}\}$ follow a similar travel like $T_1$. The matches between teams in $T_q$ for all $q \in \{\frac{a^2}{2}+1,\dots,a^2\}$ and the teams in $U$ that are in tour $j$, for all $j \in \{\frac{a}{2k}+1,\dots,\frac{a}{k}\}$ are also played in a similar fashion with the change that the teams in $U$ visit the teams in $T$ first and then they get visited by the teams in $T$. This way the sets $T$ and $U$ both are divided in two parts according to the teams they completed playing. Formally,
\begin{equation*}
	T_a=\bigcup\limits_{q=1}^{a^2/2} T_q \qquad T_b=\bigcup\limits_{q=a^2/2+1}^{a^2} T_q \qquad U_a=\bigcup\limits_{j=1}^{a/2k} S_j \qquad U_b=\bigcup\limits_{j=a/2k+1}^{a/k} S_j
\end{equation*}  
Afterwards by changing the roles of $T_a$ and $T_b$, matches between teams in $T_a$ and $U_b$ are arranged and similarly between teams in $T_b$ and $U_a$.
So, the main remaining part is schedule of matches between the teams in $U$ and that of the teams in $T$.

For scheduling matches between teams in $U$, the teams in $U$ are categorized in $k$ categories depending on their occurrence in the tours, i.e., for all $i \in \{1,\dots,k\}, U_i=\{u_{i,j}$ such that $j \in \{1,\dots,\frac{a}{k}\}\}$. For all $i$, the teams in $U_i$ play against the teams in $T$ at the same slots but half of them play at home and the other half play away. More specifically the teams in $U_i$ play with the teams in $T$ at exactly on the slots $(2i-1)$ to $(2a^3+2i-2)$. Keeping these busy slots in mind, schedule of a single round-robin tournament of the  teams in $U$ is designed using canonical tournament introduced by de Werra\cite{de1988some}. This is done by assigning a vertex to each of the teams in a $U_i$ in a special graph as done in the canonical tournament design. This tournament structure gives assurance that each team plays every other team exactly once and no team has a long sequence of home and away matches. Here a match between $i$ and $j$ signifies a match between a team in $U_i$ and a team in $U_j$. At the end the same tournament is repeated with changed match venues, i.e. nullification of home field advantage.

Now the only remaining part is scheduling of the matches between the teams in $T$. Let for all $t \in T, d(t)$ denote the first slot on which team $t$ plays a team in $U$ and let $T_i=\{t \in T $ such that $ ((d(t)-1)\mbox{mod k})+1=i\}$ for all $i \in \{1,\dots,k\}$. Then every $T_i$ is partitioned into $2a^2$ groups of cardinality $\frac{a}{2k}$ such that $d(t_1)=d(t_2)$ for every two members $t_1,t_2$ of the same group. For every $T_i$, the matches between teams in different groups in $T_i$ are scheduled. Among the teams in $T_i$, $\frac{a}{k}$ will always be busy playing with some teams of $U$. To handle this fact, two dummy groups $U_1$ and $U_2$, as defined before, are introduced that play with these busy $\frac{a}{k}$ teams of $T_i$. Then each group is treated as a team and again canonical tournament structure is applied only skipping the day at which the two dummy groups meet. 

For scheduling the matches between the members of two groups $g$ and $h$ of $T_i$ in $l$ rounds, where
\begin{equation*}
	g=\{g_1,g_2,\dots,g_l\} \qquad h=\{h_1,h_2,\dots,h_l\} 	
\end{equation*}
The following is done. The $i^{th}$ round contains the matches between $h_j$ and $g_{((i+j) \mod l)+1}$ for all $i,j \in \{1,2,\dots,l\}$ with game taking place at the home of $h_j$. This restricts the lengths of away trips and home stands for all the teams from being long. Then the same schedule is repeated with altered venues. Matches between the teams of some group of $T_i$ with the teams of a dummy group has already been taken care when the matches between $U$ and $T$ were designed. 

Now, for the scenario where the two dummy groups meet, two kind of matches are there which differ in length. The encounter between two groups consists of $\frac{a}{k}$ slots, while the encounter between two dummy groups, i.e. $U_1$ an $U_2$ consists of $a$ slots. The encounters between the groups of $T_i$ are scheduled in usual way using $\frac{a}{k}$ slots and the extra $\frac{(k-1)a}{k}$ slots are used to schedule matches between the teams in different $T_i$'s.

To schedule matches between the teams in different $T_i$'s, first each $T_i$ is partitioned in two equal size, namely $T_{i,1}$ and $T_{i,2}$. Now, considering each $T_{i,j}$ as a single team for all $i \in \{1,2,\dots,k\}$ and all $j \in \{1,2\}$, a canonical tournament is again applied on these teams skipping the day at which $T_{i,1}$ encounters $T_{i,2}$ for all $i$. These scenario can be achieved by properly initializing the canonical tournament. Then the same schedule is repeated with altered venues to nullify home advantage as done in all earlier canonical tournament structures.

For a team situated in one of the vertex in $E \setminus v$ to visit the teams in $v$, it has to travel at most $\frac{2(M+2)a^3}{k}$ using $\frac{a^3}{k}$ trips of length $k$. So, the total travel distance of all the teams in $E \setminus \{v\}$ to the teams in $v$ can be bounded by $\delta_2$. Where,
\begin{equation*}
	\delta_2 = \frac{2(M+2)a^4}{k}
\end{equation*} 

As all the distances between the vertices in $E \setminus \{v\}$ in $G$ is less or equal to $2k$, the travel between all the teams in $E \setminus \{v\}$ can be bounded above by $\delta_3$. Where,
\begin{equation*}
	\delta_3 =2a(a-1)k
\end{equation*}

The teams in $T$ visit each other at zero cost as they are situated at the same point. So, the total travel distance of the tournament is bounded by,

\begin{eqnarray*}
	&&\delta_1 \cdot a^3 + \delta_2 + \delta_3\\ &&=a^{3}p\left[M\left(k+3-\frac{2}{k}\right)+k^2-3k+6\right]+\frac{2(M+2)a^4}{k}+2a(a-1)k\\ &&=\delta
\end{eqnarray*}
In another way,
\begin{eqnarray*}
	\delta&&=Ma^3\left[p\left(k+3-\frac{2}{k}\right)+\frac{2a}{k}\right]+pa^3(k^2-3k+6)+2a(a-1)k+\frac{4a^4}{k}\\ &&=M\left[\frac{4a^4}{k}\right]+pa^3(k^2-3k+6)+2a(a-1)k+\frac{4a^4}{k}	
\end{eqnarray*}

This completes the first direction of the proof.

\subsection{Proof of the other direction}
For the other direction of the proof, it has to be shown that \textit{k-SAT} is not satisfiable $\imply$ a \textit{TTP-k} schedule of total travel distance less than $\delta$ is not possible. Here it is assumed that there is no satisfying assignment of variables of the \textit{k-SAT}. In the forward direction of the proof explained above, it is shown that the proposed schedule is compact and gives an optimized value of travel distance for the tours of the teams in $T$ to the teams in $U$. But the travels are designed depending on a truth assignment of $x$ variables. The travel to $c_m$ vertex goes through a $x_{i,j}$ or $\bar{x}_{i,j}$ which is there in the $m^{th}$ clause of the \textit{k-SAT} and assigned with value $1$. Also the other variable among these two covers $y_{i,j}$ and $w^s_{i,j}$s together in another tour. Let assume that, there exist optimum tours similar to the forward direction of the proof although there is no satisfying assignment of variables of the \textit{k-SAT}. So, there is an optimum path through each $c_m$ for all $m\in\{1,\dots,p\}$ that includes a vertex $x_{i,j}$ or $\bar{x}_{i,j}$ which is present in the $m^{th}$ clause of the \textit{k-SAT} expression. Now if value $1$ is assigned to each of these variables then that must end in a wrong assignment of variables as it contradicts the assumption otherwise. This imply that  $x_{i,j}=\bar{x}_{i,j}=1$ has been assigned for some $i \in \{1,\dots,t\}$ and $j \in \{1,\dots,n_i\}$. That means both $x_{i,j}$ and $\bar{x}_{i,j}$ are included in optimized tours from $v$ to $c_i$ and $c_j$, where $i,j\in\{1,\dots,p\}$. Now, it will not be possible to design a trips from $v$, that includes $y_{i,j},w^1_{i,j},\dots,w^{k-2}_{i,j}$ for all $i \in \{1,\dots,k\}$ and $j \in \{1,\dots,\frac{a}{k}\}$ using the optimized trips.  So, it is not possible to cover all the $c_m,y_{i,j},w^s_{i,j}$s along with all $x_{i,j}$s for all $i,j$ using these optimized tours. To cover all the vertices in $U$ each of the $a^3$ vertices of $T$ has to tour at least once more to the vertices of $U$. This extra tours will increase the total travel distance by at least $2 \cdot M \cdot a^3$. As the total travel distance by the teams in $U$ for matches among them is $O(a^2)$, so the tours of the teams of $T$ to those of $U$ dictates the total travel distance of the tournament. So, increase in this part will increase the total distance significantly and for $M$ being $\theta(a^5)$, the total travel distance will be more than $\delta$. This completes the other direction of the proof and the reduction. 

\bibliographystyle{unsrt}
\bibliography{TTP-k}

\end{document}